%% file: main.tex
\documentclass[11pt]{article}
\usepackage{inputenc}
\usepackage{thm-restate}
\usepackage{lineno}

\usepackage[letterpaper,top=1in,bottom=1in,left=1in,right=1in]{geometry}
\title{Tight Bounds for Sorting Under Partial Information}
\author{Anonymous Author(s)}
\date{}

\input{packages}

\input{macrosetup}

\input{mathsetup}

\input{tikzsetup}

\begin{document}
\author[1]{Ivor van der Hoog}
\author[1]{Daniel Rutschmann}
\affil[1]{Technical University of Denmark, DTU}

\thispagestyle{empty}

\maketitle

\input{abstract}

\thispagestyle{empty}

\paragraph{Funding.}
Ivor van der Hoog received funding from the European Union's Horizon 2020 research and innovation programme under the Marie Sk\l{}odowska-Curie grant agreement No 899987.
Daniel Rutschmann received funding from Carlsberg Foundation
Young Researcher Fellowship CF21-0302 “Graph Algorithms with Geometric Applications”.

\newpage
\setcounter{page}{1}

\input{introduction}

\input{preliminaries}

\input{upper_bound}

\input{entropy}

\input{lower_bound}


\appendix

\input{appendix}

\bibliography{references,zotero}

\end{document}

%% file: packages.tex
\usepackage{amsmath}
\usepackage{amssymb}
\usepackage{amsthm}
\usepackage{hyperref}
\usepackage[capitalize]{cleveref} 
\usepackage{fancyhdr}
\usepackage{booktabs}
\newcommand{\ra}[1]{\renewcommand{\arraystretch}{#1}}
\usepackage{graphicx}
\usepackage{tikz}
\usepackage{titlesec}
\usepackage[shortlabels]{enumitem}
\usepackage{algorithm}
\usepackage{algpseudocode}

\usepackage{authblk}

%% file: macrosetup.tex
\newcommand{\pa}[1]{\left(#1\right)}
\newcommand{\bra}[1]{\left[#1\right]}
\newcommand{\set}[1]{\left\{#1\right\}}

\renewcommand{\P}[1]{\mathrm{Pr}\bra{#1}}

\newcommand{\cI}{\mathcal{I}}
\newcommand{\tO}{\widetilde{O}}

\newcommand{\ol}{\overline}

\newcommand{\scoop}{\coprod} 
\newcommand{\tc}[1]{\operatorname{T c}\pa{#1}}

\DeclareMathOperator{\Stab}{STAB}

\newcommand{\prev}{\mathrm{prev}}

\newcommand{\cardinal}{CFJJM}
\newcommand{\Cardinal}{(\cardinal)}
\newcommand{\van}{HKLS}
\newcommand{\Van}{(\van)}

  \usepackage{todonotes}
  \tikzset{notestyleraw/.append style={rectangle}}
  {%
  \begin{description}}
  {\end{description}}

%% file: mathsetup.tex
\newtheorem{claim}{Claim}
\newtheorem{lemma}{Lemma}
\newtheorem{definition}{Definition}

\newtheorem{proposition}{Proposition}
\newtheorem{theorem}{Theorem}

\algblockdefx{Indent}{EndIndent}{}[1]{}
\algnotext[Indent]{EndIndent}

%% file: tikzsetup.tex
\usetikzlibrary{backgrounds}
\tikzset{every node/.style={draw,circle,fill=black,minimum size=4pt,inner sep=0pt}}
\tikzset{based/.style={ultra thick,blue}}
\newcommand{\ppdraw}[1]{%
    \def\points{{#1}}
    \foreach \pa [count=\c,remember=\pa as \pb] in \points {
        \pgfextra{\ifnum\c=0\relax\else \draw (\pa) -- (\pb);\fi}
    }
}

%% file: abstract.tex
\begin{abstract}
    Sorting is one of the fundamental algorithmic problems in theoretical computer science.
It has a natural generalization, introduced by Fredman in 1976, called \emph{sorting under partial information}.  The input consists of:

- a ground set $X$ of size $n$,

- a partial oracle $O_P$ (where partial oracle queries for any $(x_i, x_j)$ output whether $x_i \prec_P x_j$,

for some fixed partial order $P$), 

- a linear oracle $O_L$ (where linear oracle queries for any $(x_i, x_j)$ output whether $x_i <_L x_j$,

where the linear order $L$ extends $P$)

\noindent
The goal is to recover the linear order $L$ on $X$ using the fewest number of linear oracle queries.

In this problem, we measure algorithmic complexity through three metrics: the number of linear oracle queries to $O_L$, the number of partial oracle queries to $O_P$, and the time spent (the number of algorithmic instructions required to identify for which pairs $(x_i, x_j)$ a partial or linear oracle query is performed). 
Let $e(P)$ denote the number of linear extensions of $P$. Any algorithm requires worst-case $\log_2 e(P)$ linear oracle queries to recover the linear order on $X$. 

In 1984, Kahn and Saks presented the first algorithm that uses $\Theta(\log e(P))$ linear oracle queries (using $O(n^2)$ partial oracle queries and exponential time). Since then, both the general problem and restricted variants have been consistently studied. The state-of-the-art for the general problem is by Cardinal, Fiorini, Joret, Jungers and Munro who at STOC'10 manage to separate the linear and partial oracle queries into a preprocessing and query phase. They can preprocess $P$ using $O(n^2)$ partial oracle queries and $O(n^{2.5})$ time. Then, given $O_L$, they  uncover the linear order on $X$ in $\Theta(\log e(P))$ linear oracle queries and $O(n + \log e(P))$ time -- which is worst-case optimal in the number of linear oracle queries but not in the time spent. 

We present the first algorithm that uses a subquadratic number of partial oracle queries. For any constant $c \geq 1$, our algorithm can preprocess $O_P$ using $O(n^{1 + \frac{1}{c}})$ partial oracle queries and time. Given $O_L$, we uncover the linear order on $X$ using $\Theta(c \log e(P))$ linear oracle queries and time, which is worst-case optimal. We show a matching lower bound for the prepossessing also, as we show that there exist positive constants $(\alpha, \beta)$ where for any constant $c \geq 1$, any algorithm that uses at most $\alpha \cdot n^{1 + \frac{1}{c}}$ partial oracle queries must use worst-case at least $\beta \cdot c \log e(P)$ linear oracle queries. Thus, we solve the problem of sorting under partial information through an algorithm that is asymptotically tight across all three metrics.
\end{abstract}

%% file: introduction.tex
\section{Introduction}

Sorting is one of the fundamental algorithmic problems in theoretical computer science.
It has a natural generalization, introduced in 1976~\cite{fredman_how_1976}, called \emph{sorting under partial information}.  Given are a ground set $X$ of size $n$ and a partial order $P$ on $X$. 
The $P$ is provided as a \emph{partial oracle} $O_P$ that for any ordered query pair $(x_i, x_j)$ can output in constant time whether $x_i \prec_P x_j$.  There exists a \emph{linear oracle} $O_L$ which specifies an unknown linear order $L$ on $X$ where $L$ extends $P$. 
$O_L$ answers queries of the form: given $x_i$ and $x_j$, is $x_i <_L x_j$? 
The goal is to recover the linear order with the fewest linear oracle queries.
Let $e(P)$ denote the number of linear extensions of $P$. 
A query strategy is a sequence of linear oracle queries that uncovers the linear order on $P$.
The \emph{log-extensions} of $P$ is equal to $\log e(P) = \log_2 e(P)$. 
Any query strategy, in the worst-case, needs to ask at least $\log e(P)$ linear oracle queries to recover the linear order on $P$.

\paragraph{Previous Work.}
Fredman introduced the problem in 1976 \cite[TCS'76]{fredman_how_1976}. He devised a query strategy that asks $\log e(P) + O(n)$ linear oracle queries.
This falls short of the desired $O(\log e(P))$ bound
when $\log e(P)$ is sublinear. In 1984, Kahn and Saks \cite[Order'84]{kahn_balancing_1984} showed
that there always exists a query which reduces the number of linear extensions by a
constant fraction; yielding a query strategy that uses $O(\log e(P))$ linear oracle queries.
 Kahn and Kim \cite[STOC'92]{kahn_entropy_1992} propose studying this problem entirely as an algorithmic problem, where 
 in addition to counting  the number linear oracle queries used, one also counts the number of algorithmic instructions required to find the next linear oracle query in the query strategy. 
 They note that the algorithm by Kahn and Saks \cite{kahn_balancing_1984} has, in addition to its $O(\log e(P))$ linear oracle queries, exponential running time. This is because in essence their algorithm counts the linear extensions of $P$, which is a \#P-complete~\cite{brightwell_counting_1991} problem. 

Kahn and Kim 1992 \cite{kahn_entropy_1992} propose the first polynomial time algorithm that performs $O(\log e(P))$ linear oracle queries.
Their key insight is to approximate $\log e(P)$ by using the graph entropy of the incomparability graph (the undirected graph on $X$ which includes the edge $\set{x_i, x_j}$ whenever they are incomparable in $P$). 
Their bottleneck is to compute this approximation in polynomial time via the ellipsoid algorithm.
Cardinal, Fiorini, Joret, Jungers and Munro \Cardinal~\cite[STOC'10]{cardinal_sorting_2010} further approximate
the graph entropy by computing a greedy maximum chain decomposition on the digraph induced by $P$. This way, they describe an $O(n^{2.5})$ algorithm
that asks $O(\log e(P))$ linear oracle queries. Moreover, their algorithm can be separated into a preprocessing
and a query phase: After spending $O(n^{2.5})$ time preprocessing $P$,
the entire interaction with the linear oracle can be
done in $O(\log e(P) + n)$ time using $O(\log e(P))$ linear oracle queries. See~\cref{tabl:upper_bounds}.

\begin{table*}[b] \ra{1.2}
	\centering
    \begin{tabular}{@{}ll@{}c@{}ll@{}cl@{}}
        \toprule
        \multicolumn{2}{@{}l@{}}{Preprocessing Phase} & \phantom{abcd} & \multicolumn{2}{@{}l@{}}{Query Phase} & \phantom{abcd} & Source \\
        \cmidrule{1-2} \cmidrule{4-5}
        Queries & Time && Queries & Time && \\
        \midrule
        $O(n^2)$ & $O(n^2)$ && $\log e(P) + O(n)$ & exponential && \cite[TCS'76]{fredman_how_1976}\\
        $O(n^2)$ & $O(n^2)$ && $O(\log e(P))$ & exponential && \cite[Order'84]{kahn_balancing_1984}\\
        $O(n^2)$ & $O(n^2)$ && $O(\log e(P))$ & $O(n^6)$ && \cite[STOC'92]{kahn_entropy_1992}\\
        $O(n^2)$ & $O(n^{2.5})$ && $O(\log e(P))$ & $O(\log e(P) + n)$ && \cite[STOC'10]{cardinal_sorting_2010}\\
        \midrule
        $O(n^{1+1/c})$ & $O(n^{1+1/c})$ && $O(c \log e(P))$ & $O(c \log e(P))$ && \cref{theo:upperbound}\\
        $\Omega(n^{1+1/c})$ & $\Omega(n^{1+1/c})$ && $\frac{1}{73} c \log e(P)$ & $\Omega(c \log e(P))$ && \cref{theo:lowerbound}\\
        \bottomrule
    \end{tabular}
\caption{\label{tabl:upper_bounds} Summary of previous work and our results.}
\end{table*}

\paragraph{Related work.}
Van der Hoog, Kostityna, L\"{o}ffler and Speckmann \Van~\cite{van_der_hoog_preprocessing_2019} study the problem in a restricted setting where the partial order $P$ is induced by a set of intervals (where $(x_i, x_j)$ are incomparable whenever their intervals $([a_i, b_i], [a_j, b_j])$ intersect and where otherwise $x_i \prec_P x_j$ whenever $b_i < a_j$). 
They present an algorithm that can preprocess the set of intervals in $O(n \log n)$ time such that they can compute the sorted order on $X$ in $O(\log e(P))$ linear oracle queries and algorithmic instructions.
CFJRJ~\cite{cardinal_efficient_2010} study a variation of this problem. 
Let $S = (s_1, \ldots, s_n)$ be a set of $n$ elements and $P$ be a partial order on $S$. 
Let $X = (x_1, \ldots, x_n)$ be a set of $n$ elements and $O_L$ be a linear order oracle on $X$. 
Find a permutation $\pi$ between $S$ and $X$ such that for all $(i, j)$: $s_i \prec_P s_j \rightarrow \pi(s_i) <_L \pi(s_j)$,  by asking questions of the form: ``is $x_i <_L x_j$?”. 
They show an algorithm using $\textsc{opt} + o(\textsc{opt}) + O(n)$ linear oracle queries. 
 Roychoudhury and Yadav~\cite{roychoudhury2022efficient} study the problem of sorting a set into a tree. 
Given is a partial order $P$ on $X$ where $P$ can be induced from a tree $T_P$ where all edges point towards the root. Given the partial oracle $O_P$, the goal is to construct the tree $T_P$ using only queries of the form: ``is $x_i \prec_P x_j$?”
They show a randomized algorithm that produces the tree $P$ in $O(d n)$ expected time, and a deterministic algorithm to output the sorted order of $X$ in $O(\omega n + n \log n)$ time. Here $d$ is the maximal degree and $\omega$ is the width of the tree. 

We note that recently and independently, HHIRT~\cite{haeupler2024fastsimplesortingusing} also present a subquadratic algorithm for sorting under partial information.
Their input is different from their paper, as they assume for input a directed $n$-vertex graph $G$ with $m$ edges which induces a partial order $P$. Their algorithm runs in $O(m + n + \log e(P))$ time and uses $O(\log e(P))$ linear oracle queries and is faster when $m$ is near linear. 
Their approach cannot separate inspecting the partial order (graph, in their case) and the linear order queries. 

\paragraph{Contribution. }
 We devise the first subquadratic time algorithm that performs
$O(\log e(P))$ linear oracle queries
for the problem of sorting under partial information. 
Similar to \Cardinal~\cite{cardinal_sorting_2010}, we achieve separate preprocessing and query phases.
For every  $c \ge 1$ we show how to preprocess $P$ in $O(n^{1+\frac{1}{c}})$ time, such that given the linear oracle $O_L$ we can compute the sorted order on $X$ in $O(c \log e(P))$ time.
We show that this result is asymptotically tight for constant $c$. Specifically, we show that there exists constants $(\alpha, \beta)$ where for any constant $c \geq 1$ any algorithm with at most $\alpha \cdot n^{1 + \frac{1}{c}}$ preprocessing time and queries requires at least  $\beta \cdot c \log e(P)$ linear oracle queries and time to recover the sorted order on $X$. 
While we use entropy-based arguments in our proofs,
the algorithm itself and the lower bound construction are purely combinatorial. We prove:

\begin{restatable}{theorem}{main} \label{theo:upperbound}
    For every constant $c \ge 1$, there is a deterministic algorithm that first preprocesses $P$ in $O(n^{1+1/c})$ time,
    and then asks $O(c \log e(P))$ linear oracle queries in $O(c \log e(P))$ time
    to recover the linear order on $X$, represented as a leaf-linked tree on $X$.
\end{restatable}

The constant $c$ in the theorem indicates a trade-off between
the precomputation phase and the recovery phase.
We complement our algorithm by an decision-tree lower bound which shows
that this trade-off is a fundamental property of sorting under partial information:

\begin{restatable}{theorem}{lowerbound_small} \label{theo:lowerbound}
    Let $c$ be a constant greater or equal to $3$.  There exists no deterministic algorithm that for all partial orders $P$ can preprocess $P$ in $\frac{1}{48} n^{1 + \frac{1}{c}}$ partial oracle queries, such that given the oracle it can recover the linear order on $X$ in at most $\frac{1}{73} c \log e(P)$ exact oracle queries.
\end{restatable}

In fact we show an even stronger lower bound, ruling out even the existence of more efficient Monte Carlo algorithms: 

\begin{restatable}{theorem}{lowerbound}\label{theo:lowerbound_full}
    For $c \geq 3$ there exists a family $\mathcal{F}$ of pairs $ (P, L)$ of partial orders and linear extensions such that
    there exists no algorithm that for more than a quarter of the $(P, L) \in \mathcal{F}$
    can preprocess $P$ in $\frac{1}{48} \cdot n^{1 + \frac{1}{c}}$ partial oracle queries and 
    recover $L$ in $\frac{1}{73} \cdot c \log e(P)$ exact and partial oracle queries.
\end{restatable}

\paragraph{Key Ideas.} 
A partial order $P$ may be viewed as a directed acyclic graph (DAG) on $X$ where there exists an edge
from $x_i$ to $x_j$ whenever $x_i \prec_P x_j$. Conversely, a DAG $D$ induces a partial order via its transitive closure $\tc{D}$
which contains an edge from $x_i$ to $x_j$ whenever there is a path from $x_i$ to $x_j$ in $D$ and $x_i \ne x_j$. A chain in $P$ is any directed path in $P$ and an antichain of size $k$ is any set of $k$ pairwise incomparable elements in $P$.  

\Cardinal~\cite{cardinal_sorting_2010} construct a greedy chain decomposition $C_0, C_1, \ldots, C_w$ of $P$ (obtained by iteratively deleting maximum directed paths from $P$).
Denote $C = \{C_i \}_{i=1}^w$ and by $E_0$ the set of all edges  in $P$ between vertices in $C_0$ and $C \setminus C_0$.
They note that this creates a partial order $P' = \tc{E_0 \cup \coprod_{i=0}^{w} C_i}$ and they prove that $\log e(P') \in O(\log e(P))$ (and thus it suffices to only consdider $P'$ after preprocessing). 
In preprocessing, they spend $O(n^{2.5})$ time and $O(n^2)$ partial oracle queries to compute the greedy chain decomposition $C_0, C_1, \ldots, C_w$.
They additionally preprocess $E_0$ in $O(n^2)$ time and partial oracle queries. 

Given the linear oracle $O_L$, they pairwise merge all chains in $\{ C_i \}_{i=1}^w$ into a chain $C^*$ in $O(\log e(P))$ time and queries through a merge-sort like procedure.
The result is a partial order $P''$ where $P'' = \tc{E_0 \cup (C_0 \sqcup C^*)}$. 
They merge $C_0$ and $C^*$, using the edges in $E_0$, in $O(n+ \log e(P))$ time using an advanced merging algorithm
closely tied to the entropy of bipartite graphs.

Computing a maximum chain in a DAG requires $\Omega(n^2)$ partial oracle queries (Appendix~\ref{app:chains}).
To avoid spending at least quadratic time, we deviate from the approach of \Cardinal~by constructing a considerably sparser graph, and by computing only \emph{approximately maximal} chains:

\begin{restatable}{theorem}{majk} \label{theo:majk}
    There is an algorithm that, given a partial oracle $O_P$ where $P$ contains an (unknown)
    maximum chain of length $n - k$, computes in $O(n \log n)$ queries and time a chain of length $n - 2 k$.
\end{restatable}

Using the above theorem as a backbone, we generalize the construction of an alternative partial order $P'$ by computing a decomposition of $w$ chains of $P$ and $m$ antichains of $P$ where antichains have size at least $w+1$. 
We show that large antichains are also useful for sorting $X$ as we show:

\begin{restatable}{lemma}{antiext}\label{lemm:antiext}
     Let a poset $(P, X)$ contain $m$ disjoint antichains of size at least $w$, then
    $$\log e(P) \in \Omega(m \cdot w \log w).$$
\end{restatable} 

\noindent
In our preprocessing algorithm, we define a decomposition which we compute in $\widetilde{O}(n w^2)$ time:

\begin{definition}
\label{def:decomp}
    For any poset $(P, X)$ and any set $V \subseteq X$ we denote by $P - V$ the graph obtained by removing all vertices corresponding to elements in $V$ from $P$. For any integer $w$, we define a partition $C_0, A, C$ of $X$ where:
    \begin{itemize}[noitemsep, nolistsep]
        \item $C_0$ is an (approximately) maximum chain of $P$,
        \item $A = \{ A_i \}_{i=1}^m$ is a maximal set of antichains of size at least $w + 1$ in $P - C_0$.
        \item $C = \{ C_i \}_{i=1}^l$ is a greedy chain decomposition of $Y = P - ( C_0 \cup A )$. 
    \end{itemize}
    Finally, we denote by $E_0$ the set of edges between $C_0$ and $C$ in $P$ and $D' = E_0 \cup \scoop\limits_{i=0}^l C_i \sqcup \scoop\limits_{i=1}^m A_i$.
\end{definition}

We prove that: (1)  $D'$ has linear size, and (2) $P' = \tc{D'}$ satisfies $\log e(P') \in O(\log e(P))$.
To compute the decomposition of Definition~\ref{def:decomp}, we use the low-width poset sorting algorithm of
Daskalakis, Karp, Mossel, Riesenfeld and Verbin (DKMRV)~\cite{daskalakis_sorting_2009} as a subroutine to prove:

\begin{restatable}{lemma}{widthdecomp}
    \label{lemm:widthdecomp}
    There is an algorithm that with $\tO(n \, w^2)$ time and queries partitions a poset on $n$ elements into a maximal set $A$ of antichains of size $w+1$, and a set $Y$ of at most $w$ chains. 
\end{restatable}

Given our preprocessed DAG $D'$, in the query phase, we merge all chains in $C = \{ C_i \}_{i=1}^l$ into one chain $C^*$ with a procedure similar to merge-sort. 
We then merge $C_0$ and $C^*$, explicitly avoiding the advanced merging algorithm by~\cite{cardinal_sorting_2010}.
Instead, we use a neat insertion sort strategy that is based on exponential search. Not only is this procedure considerably simpler, it requires no further preprocessing and speeds up the running time of this step from
$O(\log e(P) + n)$ to $O(\log e(P))$, which is worst-case optimal even with infinite preprocessing. 
Finally, we show that we can insert all points in $A = \{ A_i \}_{i=1}^m$ into $C_0 \cup C^*$ in $O(c \cdot \log e(P))$ time using standard binary search. This last step is the bottleneck of our query phase. In Section~\ref{sec:lowerbound} we give a matching lower bound, proving that our algorithm is asymptotically tight with respect to the number of partial order queries, linear order queries and time spent.

%% file: preliminaries.tex

\section{Preliminaries}
A (strict) \emph{partial order} $P$ on a ground set $X$ is a binary relation on $X$ that is
irreflexive, asymetric and transitive. In other words, the digraph on $X$ with an
edge $(x_i, x_j)$ for every $x_i \prec_P y_i$ is a transitively closed DAG.
$(P, X)$ is then called a \emph{poset}.
When clear from context, we omit $X$ and use $P$ to denote the poset $(P, X)$.
A \emph{linear order} (also called \emph{total order})
is a partial order with \emph{no incomparable} pairs of elements.
The DAG of a linear order contains a Hamilonian path.
A \emph{linear extension} $L$ of a partial order $P$ is a linear order $<_L$ such that
$x_i <_L x_j$ whenever $x_i \prec_P x_j$. A linear extension corresponds to a topological ordering of the DAG $P$. We assume that upper bound algorithms run on a RAM, for a pointer machine algorithm see Appendix~\ref{app:pointer}. 

\paragraph{Input and output.}
We split our algorithm into a preprocessing phase and query phase.
In the preprocessing phase, our input is some \emph{partial oracle} $O_P$ that for any ordered query pair
$(i, j)$ with distinct $i$ and $j$ returns a boolean indicating whether $x_i \prec_P x_j$.
If $x_i \not \prec_P x_j$ and $x_j \not \prec_P x_i$ then the elements are incomparable in $P$. 
In the query phase, our input is a \emph{linear oracle} $O_L$ which specifies a linear order $L$ that extends $P$.
For any query pair $(i, j)$ a linear oracle query outputs whether $x_i <_L x_j$. 
We measure the algorithmic complexity in the number of linear oracle and partial oracle queries and the \emph{time} spent  (the number of algorithmic instructions, which includes queries). 

For output, we recall three previously studied problem variants: 
Kahn and Saks~\cite{kahn_balancing_1984} and Kahn and Kim~\cite{kahn_entropy_1992} output a \emph{complementing set}: a set of directed edges $E \subset X \times X$ where  $\forall (x_i, x_j) \in E$: $x_i <_L x_j$ and the partial order $P$ (viewed as a DAG) together with $E$ contains a Hamiltonian path. The complementing set has size $O(\log e(P))$ and can therefore be outputted in $O(\log e(P))$ time. 
Cardinal~\cite{cardinal_sorting_2010} have a query algorithm that runs in $O(\log e(P) + n)$ time and may thus output a linked list of $X$ in its sorted order. 
 \Van~\cite{van_der_hoog_preprocessing_2019} study the partial order sorting problem for the restricted case where $P$ can be induced from a set of intervals (where two elements are incomparable whenever the intervals overlap). To avoid trivial $\Omega(n)$ lower bounds for reporting a linked list, they propose to output either a linked list or a balanced binary tree on $X$ (in its sorted order), by providing a pointer to either the head of the linked list or the root of the binary tree. Consider the example where $P$ is a linear order and thus $\log e(P) \in O(1)$. This definition of output allows one to preprocess $P$ in $O(n \log n)$ time and output the sorted order of $X$ in $O(1)$ time in the query phase. 
We preprocess any partial oracle $O_P$ to (given $O_L$) output either a complementing set, a pointer to a linked list on $X$, or the root of a balanced binary tree on $X$ in its sorted order.

\paragraph{Decision Trees (DT).}
For our lower bound, we define decision trees that are a strictly stronger model of computation
than the RAM and pointer machine model.
In the preprocessing phase, a DT is a binary tree $T$ that receives as input an arbitrary partial oracle $O_P$. Each inner node of $T$ stores an ordered pair
of integers representing the query: ``$x_i \prec_P x_j$?''. For each inner node, its left branch corresponds
to $x_i \not \prec_P x_j$ and its right branch to $x_i \prec_P x_j$. For a leaf $\ell$ in $T$,
denote by $E_\ell$ the set of all directed edges $(x_i, x_j)$ where there exists a node $v$ on
the root-to-leaf path to $\ell$ that stored $x_i \prec_P x_j$ and where the path continues right. Each leaf $\ell$ stores the  DAG $Q_\ell = \tc{E_\ell}$.

In the query phase, a DT is a family of trees $\{ T_Q \}$ with a tree for each transitively closed DAG $Q$.
Each tree $T_Q$ receives as input an arbitrary linear oracle $O_L$ where the linear order $L$ extends $Q$.  
It is a binary tree where each inner node corresponds to a query ``$x_i <_L x_j$?''. For a leaf $s$ in the tree, denote by $E_s$ the edge set of
all directed edges $(x_i, x_j)$ where there exists a node $v$ on the root-to-leaf path to $s$ that stored $x_i <_L x_j$ and where the path continues right.
For any DT in the query phase, we require for each leaf $s$ that the graph $Q \cup E_s$ contains a Hamiltonian path.

Any deterministic and correct pair of  preprocessing and query algorithm on a RAM or pointer machine has a corresponding pair of DTs $(T, \{ T_Q\} )$. The worst-case number of queries performed by the preprocessing algorithm is lower bounded by the height of $T$. The worst-case number of queries performed by the query phase algorithm is lower bounded by $\max_{\ell \in T} \textsc{height}(T_{Q_\ell})$.

\paragraph{Chains and antichains.}
A \emph{chain} in $P$ is a set $C \subseteq X$ of pairwise \emph{comparable} elements.
Note that $P$ defines a linear order on $C$ and that $C$ thus corresponds to a directed path in $P$. 
The \emph{length} of a chain $C$ is $|C|$.
A chain is \emph{maximum} if its length is maximal among all chains.
A \emph{greedy chain decomposition} as introduced by \Cardinal~\cite{cardinal_sorting_2010} is
a partition $C$ of $X$ into chains $C_1, \dots, C_k$ such that $C_i$ is a maximum chain
on $(P, X \setminus (C_1 \cup \dots \cup C_{i-1}))$.
The complexity of computing a greedy chain decomposition is $O(n^{2.5})$ and $\Omega(n^2)$.
We may associate a partial order $P' = \scoop_{i=1}^{k} C_k$ to $C$
by taking the total order induced by $P$ on each $C_i$; making elements of different chains incomparable.
An \emph{antichain} in $P$ is a set of pairwise \emph{incomparable} elements.
The \emph{width} of $P$ is the size of the largest antichain.
By Dilworth's theorem, this equals the minimum number of chains needed to cover $P$.

\paragraph{Log-extensions, graph entropy, and the incomparability graph.}
For any poset $P$ we denote by $e(P)$ the number of linear extensions of $P$. We define the \emph{log-extensions} of $P$ to be $\log e(P)$ (henceforth, all logarithms are base $2$). 
To approximate the log-extensions of $P$, we define the \emph{incomparability graph} $\ol{P}$ of $P$ as the undirected graph on $X$ with edges
between any pair of incomparable pairs of elements.
The \emph{(graph) entropy} of an undirected graph $G$ was introduced by K\"orner 1973~\cite{korner1973coding}.
For a set $S \subseteq V(G)$, let $\chi^S$ denote its \emph{characteristic vector},
with value $1$ for every vertex in $S$ and $0$ otherwise. Let:
\[
\Stab(G) = \textnormal{Convex Hull}\set{\chi^S \ \Big| \ S \text{ is an independent set in } G},
\]
then the \emph{entropy} is defined as:
\[
H(G) = \min_{\lambda \in \Stab(G)} -\frac{1}{n} \sum_{u \in V(G)} \log(\lambda(u)).
\]
Note that $0 \le H(G) \le \log n$ and $H(G') \le H(G)$ for any spanning subgraph $G' \subseteq G$.
The graph entropy allows us to approximate the log-extensions of $P$:
\begin{lemma}[Theorem~1.1 in~\cite{kahn_entropy_1992}, improved to Lemma~4 in~\cite{cardinal_sorting_2010}]
\label{lemm:eentropy}
For any partial order $P$:
\[
\log e(P) \le n \cdot H(\ol P) \le 2 \log e(P)
\]
\end{lemma}

\noindent
Greedy chain decompositions provide yet another approximation for the log-extensions of $P$:

\begin{theorem}[Theorem~2.1 in~\cite{cardinal_efficient_2010}, rephrased as in Theorem~1 in \cite{cardinal_sorting_2010}] \label{theo:greedy}
If $C$ is a greedy chain decomposition of $P$, then for every $\varepsilon > 0$
\[
H(\ol{C}) \le  (1+\varepsilon) H(\ol P) + (1+\varepsilon) \log(1+1/\varepsilon),
\]
Setting $\varepsilon = 1$ yields $
H(\ol C) \le  2 H(\ol P) + 2.$
\end{theorem}

\noindent
We note that if $C = (C_1, \ldots, C_l)$ then a short computation shows that:
\[
H(\ol C) = -\sum_{i=1}^{l} \frac{|C_i|}{n} \log \Big(\frac{|C_i|}{n}\Big) = \sum_{i=1}^{l} \frac{|C_i|}{n} \log \Big(\frac{n}{|C_i|}\Big).
\]

Let $C_0, \ldots, C_l$ be a greedy chain decomposition of $P$.
We can store $C_0, \ldots, C_l$ using only $O(n)$ space. 
For any partial order that is a collection of chains, there exists an asymptotically optimal merging strategy corresponding to Huffman codes:

\begin{lemma}[Theorem 2~\cite{cardinal_sorting_2010}] \label{lemm:huffmanmerge}
    Let $Y \subseteq X$. Let $P'$ be the partial order induced by $P$ on $Y$. Let $C = \set{C_i}_{i=1}^l$ be a greedy chain decomposition of $P'$.
    For any linear oracle $O_L$ extending P, consider the linear oracle $O_{L'}$ where $L'$ is a linear order on $Y$ induced by $L$.
    Given $O_{L'}$, repeatedly
    merging the two smallest chains in $\set{C_i}_{i=1}^l$ (until we uncover $L'$ on $Y$) takes $O(|Y| \cdot (1 + H(\ol P')) = O(\log e(P') + |Y|)$ linear oracle queries and time.
    In addition, $O(\log e(P') + |Y|) \subseteq O(\log e(P) + |Y|)$.
\end{lemma}

\noindent
Finally, we use the following entropy argument by \cardinal~\cite{cardinal_sorting_2010}:
\begin{lemma}[Lemma~5 in~\cite{cardinal_sorting_2010}] \label{lemm:longentropy}
    Let $C$ be a maximum chain in $P$, then
    $
        n - |C| \in O(n \cdot H(\ol P)) = O(\log e(P)). 
    $
\end{lemma}

\paragraph{Finger search trees. }
A \emph{finger search tree} is a data structure that maintains a linearly ordered set plus
an arbitrary number of \emph{fingers} that each point to an element of the set.
Let $n$ denote the size of the sorted set.
We require the following operations:
\begin{itemize}[noitemsep]
\item (Build) Build a new tree on $n$ elements in $O(n \log n)$ time.
\item (Binary Search) Given a new element $y$, create a new finger at the largest element in the tree that is less than $y$ in $O(\log n)$ time.
\item (Exponential Search) Given a finger and a new element $y$, move the finger to
largest element in tree that is less than $y$.
This takes $O(1 + \log d)$ time where $d$ is the number of elements between the old and new finger.
\item (Insert) Given a finger, insert a new element directly after the finger in $O(1)$ \emph{amortized} time, i.e. $k$ insertions take $O(k)$ time in total, independent of any other operations performed.
\end{itemize}
The space usage is $O(n + f)$ where $f$ is the current number of fingers.
Finger search trees can be implemented as level-linked red-black trees \cite{tarjan_on_1988}.
Some care has to be taken to guarantee $O(1)$ amortized insertion time after building.
More complicated worst-case variants are described in \cite{brodal_optimal_2002}.

%% file: upper_bound.tex
\section{Algorithm description}
We give an overview of our algorithm. 
Our algorithm works in two phases. First, there is a preprocessing phase during which the algorithm
interacts with the partial oracle $O_P$. After that, the algorithm enters the query phase and may interact with the linear oracle $O_L$.
Our algorithm takes a parameter $c \ge 1$ specifying the trade-off between the running time of the preprocessing phase and query phase. Specifically, our preprocessing phase uses $O(n^{1  + \frac{1}{c}})$ time and our query phase takes $O(c \log e(P))$ time.
The respective algorithms are described in \cref{algo:preproc} and \cref{algo:query}.
If $O(\log e(P))$ is  (super)-linear, then after the query phase we may report all points in $X$ in their sorted order in $O(\log e(P))$ time. Otherwise, we either output a complementing set of $P$, or point to the root of a leaf-linked balanced binary tree that stores the set $X$ in its sorted order. We describe how to output our tree. The complementing set may be obtained by recording all successful linear oracle queries when constructing the tree. 

\subparagraph{Algorithmic description.}
During preprocessing, we compute the decomposition of Definition~\ref{def:decomp}. 
We first remove an (approximately) maximum chain $C_0$ from $P$ and store it in a balanced finger search tree $T_0$ (in order along $C_0$). 
For some smartly chosen $w$, we compute a maximal set $A$ of antichains of size at least $w+1$ of $P - C_0$.  
This gives a set $Y = P - (C_0 \cup  A)$ which we obtain as a set of $w$ chains. 
We then partition $Y$ into a greedy chain decomposition $C = \{ C_i' \}_1^\ell$.
Each $y \in C$ receives a finger to the maximum $x \in T_0$ with $x \prec_P y$. each $y \in C$ also receives a finger to the minimal $x \in T_0$ with $y \prec_P x$. 

At query time, we merge the chains in $C$ via \cref{lemm:huffmanmerge} into one chain $C^*$. We then merge $C^*$ with $C_0$. We do this by starting for all $y \in C^*$ an exponential search in $T_0$ from the fingers of $y$. 
Finally, we insert all $y \in A$ into $T_0$ through a regular binary search. 

The main difficulty in our approach is that computing a maximum chain $C_0$ or a greedy chain decomposition takes $\Omega(n^2)$ time.
We avoid this lower bound in two ways: first, we compute only an approximately longest chain $C_0$. 
Second, after we have deleted large antichains from $P$, we show that we may assume that $Y  = P - (C_0  \cup A)$ is a set of at most $w$ chains. This allows us to compute a greedy chain decomposition of $Y$ in $O(n w^2)$ time.  
Having merged $C$ into one long chain $C^*$, we need to merge $C^*$ and $C_0$ into a single chain. To avoid the quadratic preprocessing of $C_0$, and linear query time from the query algorithm by  \Cardinal~\cite{cardinal_sorting_2010}, we rely upon exponential search.

\begin{algorithm}[h]
\caption{Preprocessing. Input: partial oracle $O_P$, $c \ge 1$}\label{algo:preproc}
\begin{algorithmic}[1]
    \State $C_0 \gets \text{approximate maximum chain in } P$ \Comment{Via \cref{theo:majk}}
    \State $T_0 \gets \text{balanced finger search tree on } C_0$
    \State $w \gets \frac{1}{2} n^{1/(3c)}$
    \State $A \gets $ maximal set of antichains of size at least $w + 1$ of $P - C_0$.\Comment{Via \cref{lemm:widthdecomp}}
    \State $Y \gets P -  ( C_0 \cup A )$, obtained as a set of $w$ chains \Comment{Via \cref{lemm:widthdecomp}}
    \State $C \gets $ greedy chain decomposition of $Y$ \Comment{Via ~\cref{lemm:greedylowwidth}} 
    \State $E_0 \gets \emptyset$
    \ForAll{$y \in C$}
        \State $y$ gets fingers to $\alpha_y  = \arg \max\{x \in C_0 \ | \ x \prec_P y \}$ and $\beta_y = \arg \min\{x \in C_0 \ | \ y \prec_P x \}$
        \State $E_0 = E_0 \cup \{ (\alpha_y, y)$, $(y, \beta_y) \}$  
    \EndFor
    \State $P' =  \Big(C_0 \sqcup C \sqcup A\Big) \cup E_0$
    \State \Return $(P', T_0, C, A)$. 
\end{algorithmic}
\end{algorithm}

\begin{algorithm}[h]
\caption{Query Phase. Input: tree $T_0$, chain decomposition $C$, antichains $A$, oracle $O_L$}\label{algo:query}
\begin{algorithmic}[1]
\State $C^* \gets \text{chainMerge}(C)$ \Comment{Via \cref{lemm:huffmanmerge}}
\State $y_{\prev} \gets \text{nil}$
\ForAll{$y \in C^*$ in order}
    \State $\ell_y \gets \max_L(y_{\prev}, \alpha_y)$
    \State Insert $y$ into $T_0$ via exponential search starting at $\ell_y$
    \State $y_{\prev} \gets y$
\EndFor
\ForAll{$y \in A$}
    \State Insert $y$ into $T_0$ with binary search.
\EndFor
\State \textbf{return} $T_0$
\end{algorithmic}
\end{algorithm}

\section{Analysing the Preprocessing Phase}

We show that our preprocessing uses at most $O(n w^2)$ time and queries by showing \cref{theo:majk} and  \cref{lemm:widthdecomp,lemm:greedylowwidth} -- whose application dominate the running time of \Cref{algo:preproc}.
In the next section, we analyse the query phase.
We begin by showing \cref{theo:majk}:

\majk*

\begin{proof}
    Consider an algorithm that iteratively considers two elements $y_1, y_2$ that are incomparable in $P$ and removes both of them, until no such pair exists.
    The remaining set of vertices in $P$ is per definition a chain $C$. We claim that if the longest chain in $P$ has a maximum length of $n - k$ then $C$ has a chain of length of at least  $n - 2k$.
Indeed, let $C^*$ be a maximum chain of length $n-k$. Whenever $y_1, y_2$ are incomparable,
at least one of them lies outside $C^*$. Thus, our procedure deletes at most $k$ pairs of elements,
implying $|C| \ge n - 2k$.

A straight-forward implementation of the above procedure takes $\Theta(n^2)$ time and partial oracle queries. To speed this up to $O(n \log n)$, we use a sorting algorithm
to find incomparable pairs. Indeed, run merge sort on $X$ using the partial oracle $P$. Whenever the merge step encounters two
incomparable elements, delete both of them from $X$ and continue with the merge step as usual. When merge sort
terminates, the remaining elements are linearly sorted and hence form a chain.
\end{proof}

\widthdecomp*

\begin{proof}
    (DKMRV)~\cite{daskalakis_sorting_2009}
    present an algorithm to partition posets of width $\leq w$ into $w$ chains.
    In Appendix~\ref{app:poset}, we adapt the core subroutine of  their algorithm
    (specifically, their `peeling' procedure)
    to instead iteratively find and delete antichains of size $w+1$. 
    This facilitates a divide-and-conquer approach that arbitrarily splits $P$
    into two parts and recursively partitions them into $w$ chains, plus some large antichains.
    To merge the two parts, we use our new subroutine to extract antichains of size $w+1$.
    After greedily extracting all large antichains, the remaining elements
    may be partitioned into $w$ chains by Dilworth's theorem.
\end{proof}

\noindent
Finally, given the output of \cref{lemm:widthdecomp}, we show:

\begin{lemma} \label{lemm:greedylowwidth}
    Given a poset $Q$ of width $\le w$, represented via an oracle,
    plus a decomposition of $Q$ into $w$ chains,
    a greedy chain decomposition of $Q$ can be computed in $O(n w^2)$ time and queries.
\end{lemma}

\begin{proof}
    A greedy chain decomposition can be obtained by repeatedly finding a maximum chain
    and deleting its vertices from the $w$ chains.
    Finding a single maximum chain corresponds to a longest path problem in the corresponding DAG.
    There, we only need to consider the edges from each element to its successor
    in all other chains, plus the edges on the paths corresponding to the chains.
    This sparsifies the DAG to contain $O(n w)$ edges
    and a longest path can be found in $O(n w)$ time.
    A similar idea was used in the \textsc{ChainMerge} structure
    of \cite{daskalakis_sorting_2009}.
    Deleting the vertices takes $O(n)$ time.

    As $Q$ has width $\le w$, the maximum chain has length $\ge n/w$,
    hence the number of remaining vertices decreases geometrically by a factor $(1-1/w)$ after each iteration.
    Thus, the total running time is $O(n w \cdot 1 / (1-(1-1/w))) = O(n w^2)$.
    (Note that the greedy chain decomposition might contain \emph{more} than $w$ chains in general.)
\end{proof}

\noindent
We can combine our observations to upper bound our preprocessing running time:

\begin{claim}
\label{clai:runningtime}
The running time of \cref{algo:preproc} is $O(n^{1+1/c})$.
\end{claim}
\begin{proof}
    By \cref{theo:majk}, constructing $C_0$ takes $O(n \log n)$ time, as does building the finger search tree.
    By
    \cref{lemm:widthdecomp}, we obtain $(A, Y)$ in $\tO(n w^2)$ time.
    Computing a greedy chain decomposition $C$ from $Y$ takes $O(n w^2)$ time by \cref{lemm:greedylowwidth}.
    Finally, for each $y \in C$ we may compute the corresponding fingers $\alpha_y$ and $\beta_y$ in $O(\log n)$ time each using binary search. The for-loop therefore takes $O(n \log n)$ time and the total running time is:
    $\tO(n \log n + n w^2) = \tO(n^{1 + 2c/3}) = O(n^{1+1/c})$.
\end{proof}

%% file: entropy.tex

\section{Analysing the Query Phase}

In this section, we show that \cref{algo:query} returns a balanced binary tree on $X$ in $O(c \cdot \log e(P))$ linear oracle queries and $O(c \cdot \log e(P))$ time.
Recall that our input consists of $(C_0, A, C)$ where  $C_0$ is an approximately maximum chain in $P$, $A$ is a maximal set of antichains of size at least $w + 1$ in $P - C_0$ and that $C$ is a greedy chain decomposition of $Y = P - ( C_0 \cup A )$.
We first note:

\begin{claim}
The running time of \cref{algo:query} is linear in the number of linear oracle queries,
i.e. if \cref{algo:query} performs $Q$ exact oracle queries,
then its running time is $O(Q)$.
\end{claim}
\begin{proof}
The running time of \cref{lemm:huffmanmerge} is linear in the number of queries performed.
The same is true for the exponential search and binary search so this implies the claim. 
\end{proof}

\noindent
In the remainder, we upper bound the number of linear oracle queries that we perform. 

\paragraph{Approach.}
\cref{algo:query} consists of three components: (1) the merging of $C$ into a chain $C^*$, (2) the first for-loop and (3) the second for-loop. 
We analyse each component separately. 
For (1), we immediately apply  \cref{lemm:huffmanmerge} which runs in $O(|Y| + \log e(P')) \subset O(|Y| + \log e(P))$ queries. 
If $\log e(P)$ is sublinear, then $|Y|$ is the dominating term. 
We show that $|Y| \in O(\log e(P))$ and thus this procedure takes $O(\log e(P))$ queries.  
For (2), we show that merging $C^*$ into $C_0$ through exponential search uses $O(\log e(P))$ queries. Contrary to previous works that try to efficiently merge chains using the graph entropy,  we use a direct counting argument.
Finally for (3) we prove \cref{lemm:antiext}: showing that we may insert all elements in the antichains $A$ into $T_0$ using binary search in $O(c \log e(P))$ queries, this is the bottleneck of our approach which gives the factor $c$.

\paragraph{(1): Sorting $C$.} We sort all chains in $C$ into one chain $C^*$ using \cref{lemm:huffmanmerge}:

\begin{lemma}
    \label{lemm:sortingC}
    Let $(C_0, A, C)$ be the input as specified at the section start. Given the linear oracle $O_L$, we can sort all chains in $C$ into a single chain $C^*$ in $O(\log e(P))$ linear oracle queries. 
\end{lemma}

\begin{proof}
The set $C$ is a greedy chain decomposition of $Y = P - ( C_0 \cup A )$.
Let $P'$ be the poset induced by $P$ on $Y$. 
    We apply \cref{lemm:huffmanmerge} to sort $C$ in $O( \log e(P) + |Y|)$ linear oracle queries and time. 
    What remains is to upper bound $|Y|$. Suppose that the longest chain in $P$ has length $k$.
    We apply \cref{lemm:longentropy} to note that  $k \in O(\log e(P))$.
    The chain $C_0$ has length at least $n - 2k$ and thus $Y$ has at most $2k$ elements. This implies $|Y| \in O(\log e(P))$ and thus concludes the lemma. 
\end{proof}

\paragraph{(2): The first for-loop.}

\begin{lemma}
    \label{lemm:first_loop}
    Let $(C_0, A, C)$ be the input as specified at the section start. The first for-loop of Algorithm~\ref{algo:query} inserts all $y \in C^*$ into $C_0$ requires $O(\log e(P))$ linear oracle queries. 
\end{lemma}

\begin{proof}
    For each $y \in C^*$, let $y$ be inserted $d_y$ positions after the pointer $\ell_y$ in $C_0$. 
    I.e., if $d_y = 1$ then $y$ is inserted immediately after $\ell_y$. The first for-loop takes $O( \sum\limits_{y \in C^*} (1 + \log d_y)) = O(|C^*| + \sum\limits_{y \in C^*} \log d_y) = O(|Y| + \sum\limits_{y \in C^*} \log d_y)$ time. In Lemma~\ref{lemm:sortingC} we showed that $|Y| \in O(\log e(P))$, thus it remains to upper bound the second term.  

    Denote by $Z \subset X$ all vertices in $C_0 \cup C^*$ and by $P'$ the partial order induced by $P$ on $Z$. 
    Denote by $L'$ the linear order induced by $L$ on $Z$ and for all $x \in Z$ by $\pi(x)$ its index in $L'$.  
    To show that $\sum\limits_{y \in C^*} \log d_y \in O(\log e(P))$ we conceptually embed $Z$ onto $[1, |Z|]$ where each $x \in Z$ is a point $\pi(x)$. 
    
    For each $y \in C^*$, consider the interval $I_y := [ \pi(\ell_y) + 1, \pi(y)]$: this is an interval that contains $d_y$ integers and per construction of $\ell_y$, these intervals are pairwise disjoint. 
    We now conceptually create a new set of points $Z''$ from $Z$. 
    Consider simultaneously picking for every $y \in C^*$ an arbitrary point $z_y \in I_y$ and translating $y$ to $z_y$ (translating all points in $[z_y, \pi(y) - 1]$ one place to the right). Since $\ell_y$ is the maximum of all incoming edges in $P'$, the result is an embedded point set $Z''$ whose left-to-right order is a linear extension $L''$ of $P'$. 
     Since all intervals are disjoint, the number of linear extensions $L''$ that we can create this way is $\prod\limits_{y \in C^*} |d_y|$. Thus, $\prod\limits_{y \in C^*} |d_y| \leq e(P') \leq e(P)$ and $\sum\limits_{y \in C^*} \log d_y \in O(\log e(P))$.
\end{proof}

\paragraph{(3): The second for-loop.}
We upper bound the running time of the second for-loop by proving \cref{lemm:antiext}. We show a new result involving the entropy of undirected graphs.
One novelty of this approach is that it avoids dealing with the transitivity condition of partial orders.

\begin{lemma} \label{lemm:antientropy}
Let $P$ be a partial order that contains a set $A$ of $m$ disjoint antichains of size at least $k$, then:
\[
n \cdot H(\ol P) \ge m k \log k.
\]

\end{lemma}
\begin{proof}
    Let $G$ be a graph on $X$ whose edges form $m$ disjoint cliques $\set{A_i}_{i=1}^{m}$ of size at least $k$,
    corresponding to the antichains in $P$. A direct computation shows $n \cdot H(G) \ge m k \log k$.
    Indeed, if $\lambda \in \Stab(G)$, then $\sum_{x \in A_i} \lambda(x) \le 1$ as an independent set in $G$ contain at most one element from $A_i$.
    By convexity of $-\log(\dots)$,
    \[
        \sum_{x \in A_i} -\log(\lambda(x)) \ge - k \log\Big(\frac{\sum_{x \in A_i} \lambda(x)}{k}\Big) \ge k \log k,
    \]
    hence
    \[
        n \cdot H(G) = \min_{\lambda \in \Stab(G)} \sum_{x \in X} -\log(\lambda(x)) \ge \min_{\lambda \in \Stab(G)} \sum_{i=1}^{m} \sum_{x \in A_i} -\log(\lambda(x)) \ge \sum_{i=1}^{m} k \log k
    \]
    \noindent
    Finally, $H(G) \le H(\ol P)$ as $G$ is a spanning subgraph of $\ol P$.
\end{proof}

\antiext*

\begin{proof}
    The proof is an immediate application of \cref{lemm:eentropy,lemm:antientropy}.
\end{proof}

\begin{lemma}
    \label{lemm:second_loop}
    The second for-loop of \cref{algo:query} performs $O(c \cdot \log e(P))$ linear oracle queries. 
\end{lemma}

\begin{proof}

Inserting a single $y \in A$ into $C^*$ with binary search uses $O(\log n)$ queries and the for-loop of \cref{algo:query} thus performs at most $O(|A| \log n)$ queries. 
    We apply \cref{lemm:antiext} with $m = |A| / (w+1)$ to note that  $\log e(P) \in \Omega(|A| \log w)$.
Note that $w = \frac12 n^{1+1/(3c)}$ and thus $\log_w(n) \in O(c)$.
We now observe that the total number of queries in the second for loop of \cref{algo:query} is
\[
O(|A| \log n) = O(\log_w(n) \cdot |A| \log(w)) = O(c \cdot \log e(P)). \qedhere
\]
\end{proof}

\noindent
By combining Claim~\ref{clai:runningtime} with~\cref{lemm:sortingC,lemm:first_loop,lemm:second_loop}, we obtain Theorem~\ref{theo:upperbound}.

%% file: lower_bound.tex
\section{Lower Bound}
\label{sec:lowerbound}

We fix any decision tree $T$ in the preprocessing phase of height at most $\frac{1}{48}n^{1+1/c}$ and any set of corresponding query algorithms $\{ T_Q \}$. 
Our strategy is to construct a family $\{ (P_i, L_i) \}_{i \in \cI}$ of pairs of a partial order with a corresponding linear extension 
with $\log e(P_i) = O(\frac{1}{c} n \log n)$. We will use this set to show a lower bound through the following concept:

\begin{definition}
    For any pair in $\{ (P_i, L_i) \}_{i \in \cI}$, if $O_{P_i}$ is the input of $T$ then the execution of $T$ ends in a leaf that stores a partial order which we denote by $Q_i$. We denote:
    \[
    E_i := \{ \textnormal{Linear orders } L_j \mid j \textnormal{ satisfies }Q_j = Q_i \}.
    \]
\end{definition}

We show that for any fixed preprocessing decision tree $T$ of height at most $\frac{1}{48}n^{1 + 1/c}$, for at least half of the pairs in $\{ (P_i, L_i) \}_{i \in \cI}$, $\log |E_i| \in \Omega(n \log n) = \Omega(c \cdot \log e(P_i))$. 
To show this, we show the following proposition:

\begin{proposition} \label{prop:lowertech}
    For every $c \ge 3$ and $n \ge 10$, there exists a family $\set{(P_i, L_i)}_{i \in \cI}$
    of partial orders $P_i$ with a linear extension $L_i$ with the following properties:
    \begin{enumerate}[1), noitemsep]
        \item $L_i$ is a linear extension of $P_i$, but not of $P_j$ for $j \ne i$.
        \item $\log e(P_i) \le \frac{1}{c} n \log n$
        \item If $T$ is a decision tree that performs $\frac{1}{48} n^{1 + 1/c}$ partial oracle queries
        during preprocessing, then for at least half the $i \in \cI$, we have $Q_i = Q_j$ for $m_i \ge 2^{\frac{1}{72} n \log n}$ distinct $j$. 
    \end{enumerate}
\end{proposition}

\begin{lemma}
    \label{lemm:implies}
    Proposition~\ref{prop:lowertech} implies \cref{theo:lowerbound_full}.
\end{lemma}

\begin{proof}

    By Condition (1), $|E_i|  = m_i$.
    Indeed if $Q_j = Q_i$ then per definition $L_j \in E_i$ and all $L_j$ are distinct.
    For each $i \in \mathcal{I}$, even if the set $\{ (P_i, L_i) \}_{i \in \cI}$ is known \emph{a priori}
    (and thus the leaves of $T_{Q_i}$ contain only Hamiltonian paths corresponding to $L_j \in E_i$), the tree $T_{Q_i}$ contains
    at least $m_i$ leaves. Moreover, at most $2^{\lfloor\log m_i\rfloor - 1} \le m_i / 2$ of these leaves have a depth $<\lfloor\log m_i\rfloor$. We may now apply Condition (3) to note that for half of the $i \in \cI$ we have $m_i \ge 2^{\frac{1}{72} n \log n}$, and in particular, $\lfloor\log m_i\rfloor \ge \frac{1}{72} n \log n$.
    It follows that for at least half of the $j$ with $L_j \in E_i$, the leaf in $T_{Q_i} = T_{Q_j}$ containing $L_j$ has depth at least $\frac{1}{72} n \log n$. Therefore, for at least a quarter of the $j \in \cI$, the query algorithm performs at least $\frac{1}{72} n \log n$ queries before outputting $L_j$.
    By Condition (2),
    $\frac{1}{72} n \log n \ge \frac{1}{72} \cdot c \log e(P_i)$. 
\end{proof}

\paragraph{The Family} Given $c$ and $n$,
let $w = n^{1/c}$ and $h = n / (2w)$. For simplicity, assume that $w$ and $h$ are both integers.
We first give our intuitive definition (see also Figure~\ref{fig:lowerbound}). 
We call the first $n/2$ elements the \emph{Essentials} and use them to build  $w$ chains of length $h$. 
Each $P_i$ can then be generated by arbitrarily placing the remaining $n/2$ elements (\emph{Leftovers}) between any two Essentials.
All Leftovers $x_\ell$ between any two 
consecutive Essentials $(x_{j + kw}, x_{j + (k+1)w})$ are linearly ordered by $\ell$.
Formally, we define $\{(P_i, L_i) \}$ to consist of all partial orders $P_i$ with the following properties 
\begin{enumerate}[a), noitemsep]
    \item $P_i$ is the disjoint union of $w$ chains.
    \item (\emph{Essentials}) for each $j \in [w]$, the elements $\{ x_j, x_{j+w}, \dots, x_{j + (h-1) w} \}$ appear on the $j$-th chain in that order.  $x_j$ is the minimum and $x_{j+(h-1) w}$ is the maximum element of the $j$-th chain (this creates $w$ chains where each chain has $h$ elements from the first $n/2$ elements of $X$). 
    \item (\emph{Leftovers}) for each $\ell > n/2$ there exists a unique pair $(j, k) \in [w] \times [0, h-2]$ with $x_{j+k w} \prec_{P_i} x_\ell \prec_{P_i} x_{j+(k+1)w}$ . All $x_\ell$ with $x_{j+k w} \prec_{P_i} x_\ell \prec_{P_i} x_{j+(k+1)w}$ are linearly ordered by $\ell$.
    \item We define $L_i$ as the concatenation of the $w$ chains (connecting $x_{j + (h-1) w}$ to $x_{j+1}$). 
\end{enumerate}

\begin{figure}
    \centering
    \includegraphics{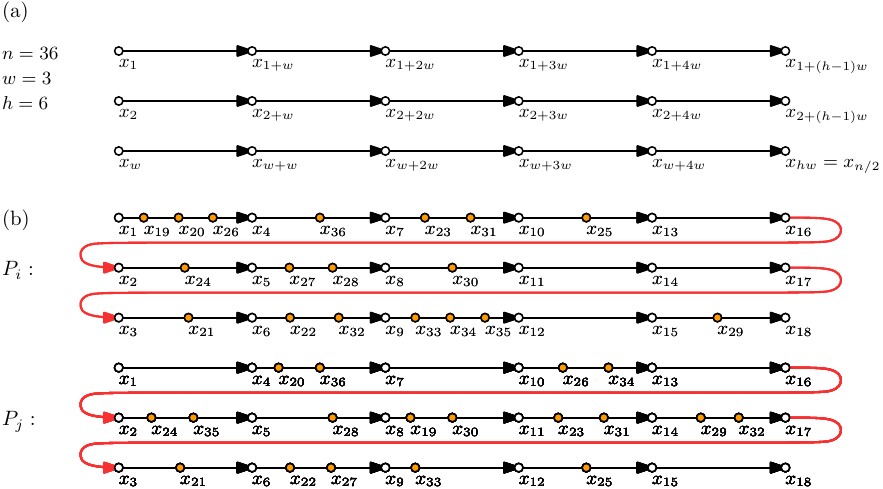}
    \caption{Our family of pairs of partial orders with linear extensions $\{ (P_i, L_i) \}$ can be constructed in two stages. (a): Step 1: partition the vertices $x_i$ for $i \leq n/2$ into $w$ equal-size chains. 
(b) For all $\ell  > n/2$, arbitrarily assign $\ell$ to lie between any two connected vertices of the previous construction. All $x_\ell$ that get assigned the same pair $(x_{j+kw}, x_{j + (k+1)w})$ are linearly ordered by $\ell$.  The corresponding linear order $L_i$ is obtained by adding the red edges.  }
    \label{fig:lowerbound}
\end{figure}

\begin{lemma}
    The family $\{ (P_i, L_i) \}_{i \in \mathcal{I}}$ has Properties 1) and 2) of Proposition~\ref{prop:lowertech}.
\end{lemma}

\begin{proof}
    First we show Property 1). 
    Given the linear extension $L_i$ of a partial order $P_i$, we can uniquely recover $P_i$ as follows: 
    split the chain $L$ into $w$ chains at  $x_1, x_2, \dots, x_w$ (making these the head of each chain).
    This uniquely recovers the $w$ chains that are $P_i$ and thus for all $i, j$: $L_i \neq L_j$. 

    To show Property 2), we note that by construction any $P_i$ is the disjoint union of $w$ chains.
    We use the definition of graph entropy, setting $\lambda \equiv 1/w$ to note that
    $H(\ol{P_i}) \le \log w = \frac{1}{c} \log n$.
    By \cref{lemm:eentropy}, $\log e(P_i) \le n H(\ol{P_i}) = \frac{1}{c} n \log n$.
\end{proof}

\paragraph{The third property} 
What remains is to show that $\{ (P_i, L_i) \}_{i \in \mathcal{I}}$ satisfies Property 3). 
Recall that $T$ is our decision tree corresponding to the preprocessing phase.
Suppose that it performed at most  $\frac{1}{48} n^{1 + 1/c}$ partial oracle queries
during preprocessing. We say a partial oracle query ``is $x_j \prec_P x_i$?''
\emph{succeeds} if the oracle returns ``yes''.

\begin{definition}
    For any decision tree $T$ of the preprocessing phase and partial order $P_i$, we define three conditions after running $T$ with $O_{P_i}$ as input:
    \begin{itemize}[noitemsep, nolistsep]
        \item An elements $x_j$ is \emph{interesting}
if $j > n/2$,
\item  an interesting element is \emph{weak}
if at most $w/4$ queries are performed on it, and
\item an interesting element is \emph{hidden} if no queries involving it succeed. 
    \end{itemize}
\end{definition}

Consider the probability space defined by taking a uniformly random $P_i$
and running $T$ with $O_{P_i}$ as input.
We show that our set of partial orders has Property 3) through showing four claims:

\begin{claim} \label{clai:manyweak}
For any $P_i$ where $O_{P_i}$ is the input of $T$, there are always at least $n/3$ weak elements.
\end{claim}
\begin{proof}
Each of the at most $\frac{1}{48} n^{1+1/c}$ queries involve at most two interesting elements. Thus, there are at most $\frac{1}{24} n^{1+1/c} / (w/4) = \frac{1}{6} n$ non-weak interesting elements.
The remaining $n/2 - n/6$ interesting elements are weak.
\end{proof}
\begin{claim} \label{clai:weakhidden}
Let $P_i$ be uniformly at random. $\forall x_j$ with $j > \frac{n}{2}$: $\P{x_j \text{ is weak and not hidden}} \le \frac{1}{4}.$
\end{claim}
\begin{proof}
If we fix the partial order of all elements except $x_j$,
then $x_j$ is equally likely to be in any of the $w$ chains, as by Condition (d) in the definition of
$\set{(P_i, L_i)}_{i \in \mathcal{I}}$, there are exactly $h-1$ options to place $x_j$ on a particular chain (in between each pair of Essentials on the chain).
Each query involving $x_j$, and $x_k$, succeeds only if $x_j$ is on the same chain as $x_k$.
Thus
\[
\P{x_j \text{ is not hidden after the first $w/4$ queries involving $x_j$}} \le \frac{w/4}{w} = \frac14.
\]
Intuitively, this corresponds to trying to find a uniformly random ball among $w$ bins
by looking at $\le w/4$ bins.
The best you can do is look at $w/4$ distinct bins, then you find the ball with probability $1/4$.
The claim now follows as [$x_j$ is weak and not hidden]
implies [$x_j$ is not hidden after the first $w/4$ queries involving $x_j$].
\end{proof}

\begin{claim} \label{clai:ehidden}
Let $P_i$ be sampled uniformly at random. Let $H$ denote the number of hidden elements after running $T$ with $O_{P_i}$ as input. Then 
$\P{H \ge n/12} \ge 1/2$.
\end{claim}
\begin{proof}
\cref{clai:weakhidden} shows that the expected number of weak non-hidden elements is $\le n/8$.
By Markov's inequality,
\[
\P{\text{\#weak non-hidden elements} \ge n/4} \le \frac12.
\]
By \cref{clai:manyweak}, there are $\ge n/3$ weak elements.
If at most $n/4$ of them are non-hidden, then at least $n/12$ of them are hidden.
\end{proof}

\begin{claim} \label{clai:hiddenoptions}
 Suppose that running $T$ with $O_{P_i}$ as input ends with $H$ hidden elements and at a leaf that stores $Q_i$.  
Then, there exist at least $(h-1)^{H}$ distinct $P_j$ compatible with the partial order $Q_i$. 
\end{claim}
\begin{proof}
If $x_k$ is a hidden element in the $\ell$-th chain in $P_i$,
then there were no queries involving both $x_k$ and another element of the $\ell$-th chain.
Thus, any $P_j$ obtained from $P_i$ by moving $x_k$ to any of the $(h-1)$ possible positions
inside the $\ell$-th chain is consistent with the query results.
(one of these $P_j$ is $P_i$ itself.)
As $x_k$ remains in the same chain, this argument applies independently to each of the $H$
hidden elements, yielding $(h-1)^H$
distinct $P_j$ in total.
\end{proof}

\paragraph{Concluding our argument.}
We now show that $\set{(P_i, L_i)}_{i \in \mathcal{I}}$ satisfies Condition (3)
of \cref{prop:lowertech}. By \cref{clai:ehidden}, for at least half the $P_i$,
we have $H \ge n/24$. For those $P_i$, by \cref{clai:hiddenoptions},
the number of $P_j$ compatible with the query results is at least
$(h-1)^{n/24} = 2^{n \log(h-1)/24}$. Finally, $c \ge 3$ and $n \ge 10$ imply $\log(h-1) \ge \log(n) / 3$.
Thus, by Lemma~\ref{lemm:implies}, we may conclude: 

\lowerbound*

%% file: appendix.tex
\section{Appendix}
\label{app:poset}

\subsection{Decomposing a poset into chains and antichains} \label{sect:peelmodify}
Recall that a poset has width $\le w$ if and only if it can be decomposed into $w$ chains.
By Dilworth's theorem, this is equivalent to there being no antichain of size $w+1$.
We now explain how to modify the  merge sort-like algorithm for sorting
posets of width $\le w$ of (DKMRV) \cite{daskalakis_sorting_2009}
to work in general posets by finding and deleting antichains of size $w+1$.
The core of the (DKMRV) algorithm is the \textsc{Peeling} procedure,
which consists of $w$ peeling iterations:
\begin{theorem}[\textsc{PeelingIteration}, Theorem~3.5 in~\cite{daskalakis_sorting_2009}] \label{theo:peeling}
    Let $P$ be a poset of width $\le w$ on a set $X$ of size $n$.
    Given an oracle $O_P$ and a decomposition of $P$ into $w+1$ chains,
    there is an algorithm that returns a decomposition of $P$ into $w$ chains
    in $O(n w)$ time.
\end{theorem}

To apply this peeling approach to posets of arbitrary width,
we introduce a new procedure that efficiently finds and deletes large antichains.

\begin{lemma}[\textsc{AntichainExtraction}]\label{lemm:antichainex}
    Let $P$ be a poset on a set $X$ of size $n$.
    Given an oracle $O_P$ and a decomposition of $P$ into $w+1$ chains,
    there is an algorithm that returns a decomposition of $P$ into $w+1$ chains
    and a \emph{maximal} set $A$ of antichains of size $w+1$,
    in $O(n w)$ time.
    (Note that the $w+1$ chains then form a decomposition of a poset of width $\le w$ by maximality of $A$.)
\end{lemma}
We defer the description of \textsc{AntichainExtraction} to later.
By alternating antichain extraction and peeling iterations $w$ times, we can turn a chain decomposition
of $P$ with $2w$ chains into one with $w$ chains and zero or more antichains of size $w+1$,
see \cref{algo:antichainpeel}.
This can be thought of as a generalization of the \textsc{Peeling} procedure of (DKMRV)~\cite{daskalakis_sorting_2009}.
By using this as a subroutine, we can then mimic merge sort, see \cref{algo:antimergesort}.
We can now prove:

\widthdecomp*
\begin{proof}
    We first analyze the running time of \cref{algo:antichainpeel}.
    A single antichain extraction or peeling iteration takes $O(n w)$ time.
    As we perform $w$ of these, this takes $O(n w^2)$ time in total.

    Now we analyze \cref{algo:antimergesort}.
    The merge sort recursion has depth $\log(n / w)$, and
    each recursion level takes $O(n w^2)$ time,
    hence the total running time is $O(n w^2 \log (n/w))$.
\end{proof}

\begin{algorithm}[h]
    \caption{\textsc{AntichainPeeling}\\
    \textbf{Input:} partial oracle $O_P$, chain decomposition $\set{C_i}_{i=1}^{2w}$ of $P$\\
    \textbf{Output:} a decomposition of $P$ into $w$ chains and zero or more antichains of size $w+1$
    }\label{algo:antichainpeel}
    \begin{algorithmic}[1]
        \State $A \gets \set{}$
        \For{$i=w, \dots, 1$}
            \State $(C_i', \dots, C_{i+w}'), A' \gets \textsc{AntichainExtraction}(O_P, (C_i, \dots, C_{i+w}))$ \Comment{\cref{lemm:antichainex}}
            \State $C_i'', \dots, C_{i+w-1}'' \gets \textsc{PeelingIteration}(O_P, (C_i', \dots, C_{i+w}'))$ \Comment{\cref{theo:peeling}}
            \State $A \gets A \cup A'$
            \State $C_i, \dots, C_{i+w-1} \gets C_i'', \dots, C_{i+w-1}''$
            \State $C_{i+w} \gets \set{}$
        \EndFor
        \State \Return $(C_1, \dots, C_w), A$
    \end{algorithmic}
\end{algorithm}

\begin{algorithm}[h]
    \caption{Chain-Antichain Merge Sort.\\
    \textbf{Input:} partial oracle $O_P$, set $X$\\
    \textbf{Output:} a decomposition of $P$ on $X$ into $w$ chains and zero or more antichains of size $w+1$
    }\label{algo:antimergesort}
    \begin{algorithmic}[1]
        \Procedure{ChainAntichainMergeSort}{$O_P, X$}
            \If{$|X| \le w$}
                \State $C \gets$ elements of $X$ as size-1 chains.
                \State \Return $C, \set{}$ \Comment{No antichains}
            \EndIf
            \State $X_1 \sqcup X_2 \gets X$ \Comment{Split $X$ into equally sized parts}
            \State $C', A' \gets$ \textsc{ChainAntichainMergeSort}($O_P, X_1$)
            \State $C'', A'' \gets$ \textsc{ChainAntichainMergeSort}($O_P, X_2$)
            \State $\triangleright$ Now $(C' \cup C'')$ is a decomposition of $X - A' - A''$ into $2w$ chains.
            \State $C''', A''' \gets$ \textsc{AntichainPeeling}$(O_P, C' \cup C'')$ \Comment{\cref{algo:antichainpeel}}
            \State \Return $C''', A' \cup A'' \cup A'''$
        \EndProcedure
    \end{algorithmic}
\end{algorithm}

\paragraph{Antichain Extraction.} Let us now show \cref{lemm:antichainex}. Since $P$ is decomposed into
$w+1$ chains, any antichain of size $w+1$ needs to contain exactly one element from each such chain.
If the maximum elements of all $w+1$ chains are pairwise incomparable, then they form an antichain.
Otherwise, there are two comparable maximum elements $x, y$ of different chains. Without loss of generality, let $x \prec_P y$,
then $x' \prec_P y'$ for any element $x'$ in the same chain as $x$, so $y$ cannot lie in an antichain
of size $w+1$. The resulting algorithm is described in \cref{algo:antiextract}.
A naive implementation which compares
all $O(w^2)$ pairs of maximum elements at each step takes $O(n w^2)$ time. This gets reduced to $O(n w)$
by comparing each such pair only once, which proves \cref{lemm:antichainex}.

\begin{algorithm}[h]
    \caption{Antichain Extraction (heavily inspired by Figure~3 of \cite{daskalakis_sorting_2009})\\
    \textbf{Input:} partial oracle $O_P$, chain decomposition $\set{C_i}_{i=1}^{w+1}$ of $P$\\
    \textbf{Output:} a decomposition of $P$ into $w+1$ chains and a \emph{maximal} set of antichain of size $w+1$
    }\label{algo:antiextract}
    \begin{algorithmic}[1]
        \For{$i=1, \dots, w+1$}
            \State $C_i' \gets$ copy of $C_i$
        \EndFor
        \State $A \gets \set{}$
        \While {all $C_i'$ are non-empty.}
            \If{there is a pair $(x, y)$ of maximum elements $x \in C_i'$, $y \in C_j'$, such that $x \prec_P y$}
                \State delete $y$ from $C_j'$
            \Else
                \State delete the maximum element from each $C_i'$
                \State add those $w+1$ elements as an antichain to $A$
            \EndIf
        \EndWhile
        \State delete the elements in $A$ from their respective $C_i$
        \State \Return $(C_1, \dots, C_{w+1}), A$
    \end{algorithmic}
\end{algorithm}

\subsection{RAM model vs Pointer Machine}
\label{app:pointer}
Our algorithms as presented run in the RAM model of computation. Here, we explain how to
match the bounds of \cref{theo:upperbound} on a pointer machine.

For the preprocessing phase, we use a standard simulation technique: Any RAM model algorithm with running time $O(T)$ may be simulated on a pointer machine in $O(T \log T)$ time by replacing arrays
with balanced search trees. Our analysis shows
that the preprocessing phase runs in $\tO(n^{1 + 2c/3})$ time in the RAM model,
hence the simulation on a pointer machine takes $\tO(n^{1+2c/3} \cdot \log n) \subseteq O(n^{1+1/c})$ time.

For the query phase, this simulation technique would only yield $O(\log e(P) \log \log e(P))$ running time.
To avoid the additional log factor, we instead note that our query phase runs on a pointer machine as is.
Indeed, the merge sort step in \cref{algo:query} can be implemented with linked lists,
and finger search trees work as-is on a pointer machine. Hence, our query phase runs in $O(\log e(P))$ on
a pointer machine too.

\subsection{A Quadratic Lower Bound for Finding Maximum Chains}
\label{app:chains}
In the introduction we state that computing an (exact) maximum chain
requires $\Omega(n^2)$ partial oracle queries. Let us now prove this.
As the first chain in a greedy chain decomposition is a maximum chain, such a lower bound
automatically applies to computing a greedy chain decomposition too.

If $P$ consists of two comparable elements, $x \prec_P y$, and all other pairs are incomparable,
then clearly any algorithm needs to spend $\Omega(n^2)$ partial oracle queries to find
the unique maximum chain $(x, y)$. However, one might argue that this example is a bit artificial
from a sorting perspective. Indeed, as $\log(e(P)) = \Theta(n \log n)$, any standard sorting algorithm
can sort $P$ in $O(\log e(P))$ linear oracle queries and time, so there is no need to compute a maximum chain.
To avoid this trivial case, we describe a more general family of partial orders.

\begin{lemma}
    For every $0 < k \le n-2$, there is a family of partial orders $P$ with a
    unique maximum chain of length $n-k$,
    such that any algorithm needs to spend $\Omega(k^2)$ time to find
    the maximum chain.
\end{lemma}
\begin{proof}
    We define $P$ as follows: Let $A$ consists of $k+2$ elements and $B$ consist of $n-k-2$ elements.
    Let $x, y \in A$ be distinct. Put $x \prec_P y$ and make all other pairs of elements in $A$ incomparable.
    Let the elements in $B$ form a chain. Put $a \prec_P b$ for every $a \in A, b \in B$.

    $P$ has a unique maximum chain $\set{x, y} \cup B$, of length $2 + |B| = n-k$.
    Any algorithm that finds this chain needs to find $(x, y)$ among all $\Theta(k^2)$
    pairs of elements in $A$. This requires $\Omega(k^2)$ partial oracle queries.
\end{proof}

If the maximum chain should be output in sorted order, then an additional $\Omega(n \log n)$ lower
bound for sorting applies, hence the total lower bound is $\Omega(n \log n + k^2)$.
Using our techniques, we can match this bound algorithmically.

\begin{lemma}
    Let $P$ be a partial order with a maximum chain of length $n - k$. Then, a maximum chain in $P$ can be computed in $O(n \log n + k^2)$ time and partial oracle queries.
\end{lemma}
\begin{proof}[Proof sketch]
    With \cref{theo:majk}, we can find a chain $C$ of length $n - 2k$ in $O(n \log n)$ time.
    For the remaining $2k$ elements, we find their predecessor and successor in $C$ in $O(k \log n)$ time.
    To find a maximum chain, it suffices to consider (1) these edges from a predecessor / to a successor in $C$, (2) the edges in the path formed by $C$, and (3) the edges between any pair of elements not in $C$. 
    This sparsifies $P$ into a DAG with $n$ vertices and $O(n + k^2)$ edges.
    There, a longest path can be found in $O(n + k^2)$ time.
\end{proof}